\documentclass{article}
\usepackage{graphicx}
\usepackage{amsfonts}
\usepackage{amsmath}
\usepackage{amssymb}
\usepackage{url}

\usepackage{dsfont}
\usepackage{fancyhdr}
\usepackage{indentfirst}
\usepackage{enumitem}
\usepackage[normalem]{ulem}
\usepackage{mathrsfs}
\usepackage{multirow}
\usepackage{diagbox}

 \usepackage[colorlinks=true,citecolor=blue]{hyperref}
\usepackage{amsthm}
\usepackage{color}

\definecolor{DarkGreen}{rgb}{0.2,0.6,0.2}

\definecolor{purple}{rgb}{0.6,0.3,0.8}

\usepackage{natbib}
\usepackage{comment}

\def\d{\mathrm{d}}
\def\laweq{\buildrel \mathrm{d} \over =}

\newcommand{\var}{\mathrm{Var}}

\newcommand{\E}{\mathbb{E}}

\newcommand{\R}{\mathbb{R}}

\newcommand{\N}{\mathbb{N}}
\newcommand{\p}{\mathbb{P}}
\newcommand{\X}{\mathcal{X}}

\newcommand{\id}{\mathds{1}}

\renewcommand{\ge}{\geqslant}
\renewcommand{\le}{\leqslant}

\renewcommand{\epsilon}{\varepsilon}
\newcommand{\esssup}{\mathrm{ess\mbox{-}sup}}
\newcommand{\essinf}{\mathrm{ess\mbox{-}inf}}
\renewcommand{\cdots}{\dots}
\theoremstyle{plain}
\newtheorem{theorem}{Theorem}
\newtheorem{corollary}{Corollary}
\newtheorem{lemma}{Lemma}
\newtheorem{proposition}{Proposition}
\theoremstyle{definition}
\newtheorem{definition}{Definition}
\newtheorem{example}{Example}

\theoremstyle{remark}
\newtheorem{remark}{Remark}

\usepackage[misc]{ifsym}

\usepackage[onehalfspacing]{setspace}

\usepackage{centernot}

\usepackage{footmisc}
\newcommand{\SEimplies}{%
  \mathrel{\rotatebox[origin=c]{315}{$\implies$}}%
}

\newcommand{\SWimplies}{%
  \mathrel{\rotatebox[origin=c]{225}{$\implies$}}%
}

\begin{document}

\title{Diversification Preferences and Risk Attitudes}
\author{Xiangxin He\thanks{Department of Statistics and Actuarial Science, University of Waterloo,  Canada. \Letter~{\scriptsize\url{g34he@uwaterloo.ca}}
} \and Fangda Liu\thanks{Department of Statistics and Actuarial Science, University of Waterloo,  Canada. \Letter~{\scriptsize\url{fangda.liu@uwaterloo.ca}}} \and Ruodu Wang\thanks{Department of Statistics and Actuarial Science, University of Waterloo,  Canada. \Letter~{\scriptsize\url{wang@uwaterloo.ca}}}}
\date{\today}
 
\maketitle

\begin{abstract}
Portfolio diversification is a cornerstone of modern finance, while risk aversion is central to decision theory; both concepts are long-standing and foundational. We investigate their connections by studying how different forms of diversification correspond to   notions of risk aversion. We focus on the classical distinctions between weak and strong risk aversion, and consider diversification preferences for pairs of risks that are identically distributed, comonotonic, antimonotonic, independent, or exchangeable, as well as their intersections. Under a weak continuity condition and without assuming completeness of preferences, diversification for antimonotonic and identically distributed pairs implies weak risk aversion, and diversification for exchangeable pairs is equivalent to strong risk aversion. The implication from diversification for independent pairs to weak risk aversion requires a stronger continuity. We further provide results and examples that clarify the relationships between various diversification preferences and risk attitudes, in particular justifying the one-directional nature of many implications.

\end{abstract}


\section{Introduction}

Diversification and risk attitudes are two of the most fundamental ideas in economics and finance. Diversification is central to portfolio selection and risk management since the seminal work of \cite{M52}, while risk aversion is fundamental to models of decision making under risk \citep{A63,P64,RS70}. Both concepts are classical and deeply embedded in practice, and yet their precise relationship is subtle. A unified understanding of how ``wanting to diversify" constrains a decision maker's risk attitude is essential both for theory---to organize the rich landscape of preference models---and for applications, where one would like to infer risk attitudes from observed diversification behavior, or to predict diversification behavior from risk attitudes.
 
\cite{D89} introduced an axiomatic notion of preference for portfolio diversification and showed that diversification is strictly stronger than strong risk aversion of \cite{RS70}, although these two concepts are equivalent under the expected utility (EU) model.  Dekel formulated diversification as a preference for any convex combination of outcomes that are already equally desirable. This approach is conceptually natural, and it is mathematically elegant as it reduces to quasi-convexity of the preferences under mild conditions, highlighted by \cite{CT02} and \cite{CL07}. Nevertheless, requiring 
diversification for all   dependence structures in the portfolio, including those without hedging effects, is  quite demanding. In practice, investors may only actively seek diversification in specific situations---for example, when combining market positions that hedge each other,   when combining insurance and reinsurance contracts, or when pooling uncorrelated assets. Outside these situations, there may be no compelling reason to treat mixing as strictly desirable, and the empirical verification of  Dekel's global notion of diversification needs to consider all types of dependence.   

This observation raises a natural question: how should diversification be formulated when decision makers only exhibit it in certain economically meaningful configurations of the portfolio risks? 
For  pairs of risks, there are four fundamental dependence structures: comonotonicity,  antimonotonicity, exchangeability, and independence; see  \cite{MFE15} for these dependence concepts in risk management.  
Diversification on antimonotonic pairs is intuitive and empirically observable, as it is common in practice for an investor to combine random payoffs that hedge each other, or to purchase an insurance policy on a potential random loss; in both cases, the decision maker prefers the combination of antimonotonic random variables. 
Diversification on independent pairs is also compelling in the context of finance and insurance, as the average of independent payoffs reduces the total payoff's variance, which is desirable as argued by \cite{M52}. 
Diversification on exchangeable pairs
reflects a tendency to combine risks that exhibit symmetry, a structure that is common for similar assets that share a common risk factor.
On the other hand, diversification on comonotonic pairs may not be appealing, as such pairs do not provide hedging or risk reduction intuitively.\footnote{These dependence concepts are also prominent in decision theory.  
Comonotonicity is fundamental to the axiomatization of the risk preferences of \cite{Y87} and the ambiguity model of \cite{S89}, 
independence is used to axiomatize risk preferences by \cite{PST20} and 
\cite{MPST24}, and antimonotonicity has  special features in sharp contrast to comonotonicity, as studied by \cite{ACV21} and \cite{PWW25}. For a pair of identically distributed (ID) risks, exchangeability includes comonotonicity, independence, and antimonotonicity as special cases.}

Our contributions are a  systematic  study of how  diversification preferences on various classes of pairs relate to the classic notions of weak and strong risk aversion; thereby, we formally connect decision theory to dependence modeling, two popular research fields. 
  Our diversification preferences are formulated on (i) all pairs of risks, (ii) ID pairs, (iii) comonotonic pairs; (iv) antimonotonic pairs, (v) exchangeable pairs, (vi) independent pairs, and (vii) intersections such as antimonotonic and ID. We weaken the assumptions of Dekel in several ways: (a) we require diversification only  for economically relevant dependence structures and pairs of risks,
 (b) we do not impose completeness or  monotonicity on the preferences, and (c) our continuity assumption, upper semicontinuity with respect to the $L^\infty$-norm, is very weak.
Each weakening 
makes our results stronger. The generalization in (a) offers new economic insights on the relationship between dependence and risk attitudes, a topic  recently explored by \cite{MMWW25} in the context of insurance. 
The generalizations in (b)--(c) are not just technical, as they allow for more important risk preferences such as the incomplete mean-variance model of \cite{M52} 
and quantile  maximizers \citep{R10}.

Our main results are first formulated on $L^\infty$, the space of bounded random variables. We find that diversification on antimonotonic and ID pairs lies strictly between weak and strong risk aversion (Theorem \ref{th:1}), 
whereas diversification on comonotonic pairs or  independent pairs  is too weak: neither implies  weak risk aversion, and they are indeed compatible with strong risk-seeking models (Propositions \ref{prop:1}--\ref{prop:2}). 
Diversification on exchangeable pairs,
or ID pairs with no restriction on the dependence, is equivalent to strong risk aversion (Theorem \ref{th:1ex}). 
We further show that under a stronger form of  continuity, called compact upper semicontinuity \citep{CM95}, diversification on independent and ID pairs lies strictly between weak and strong risk aversion (Theorem \ref{th:2}). These results highlight 
that the intuitively plausible and empirically observable property of diversification on antimonotonic (or independent) and ID pairs leads to weak risk aversion, 
and extending the property to exchangeable pairs 
gives rise to strong risk aversion.  
Figure \ref{fig:1} summarizes the main obtained implications. 
Furthermore, under mild conditions, neutrality  to any of the diversification classes above is equivalent to risk neutrality (Theorem \ref{th:3}). 
The results are generalized to $L^p$ for $p\ge 1$ through a new result (Theorem \ref{th:lln}) that can be seen as a  law of large numbers   for negatively dependent sequences \citep{L66} on $L^p$,  which may be of independent interest in probability theory.

\begin{figure}[ht]
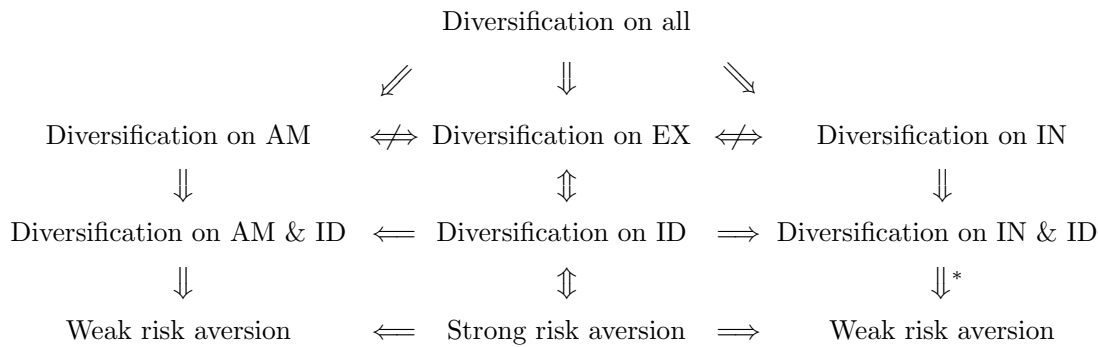


\begin{align*}  \setlength{\arraycolsep}{1pt}
\resizebox{\textwidth}{!}{$
\begin{array}{ccccc} &  & \mbox{Diversification on all} & &  \vspace{0.2cm}
\\ 
\vspace{0.2cm}
    &  \SWimplies  & \big\Downarrow &  \SEimplies & 
\\
\mbox{Diversification on AM } &\centernot\iff&\mbox{Diversification on EX }   &\centernot\iff& \mbox{Diversification on IN}   \vspace{0.2cm}
\\ 
\vspace{0.2cm}    \big\Downarrow   &    &  \big\Updownarrow  & &   \big\Downarrow
\\ 
 \mbox{Diversification on AM \& ID } &\impliedby& \mbox{Diversification on ID } &\implies& \mbox{Diversification on IN \& ID } \vspace{0.2cm}
\\
\vspace{0.2cm}  \big\Downarrow & &  \big\Updownarrow & &   ~ \big{\Downarrow^*}   
\\
 \mbox{Weak risk aversion } & \impliedby & \mbox{Strong risk aversion} &  \implies    & \mbox{Weak risk aversion}
\end{array}
$}
\end{align*}
\caption{Summary of results for risk preferences, 
where  ``AM" stands for ``antimonotonic"  (we omit ``pairs"), ``EX" stands for ``exchangeable",  ``IN" stands for ``independent", $\centernot\iff$ means incomparable, and $\Downarrow^*$ requires  compact upper semicontinuity. The converse statements of all single-direction implications do not hold for general risk preferences.}\label{fig:1}
\end{figure}

The results in the paper require substantial technical innovations. 
The proofs of the main results involve an
iterative averaging and symmetrization scheme based on antimonotonic and independent couplings, using quantile transforms and a representation of \cite{S65}.  For antimonotonic couplings, this iteration yields a sequence of payoffs with the same mean and strictly shrinking range, utilizing a technical lemma of \cite{HWWW24}. The shrinking range is important for us to  use  $L^\infty$-upper semicontinuity. 
Theorems \ref{th:1}--\ref{th:2}  generalize 
 several results in the literature, including  
\cite{D89} and \cite{CL07} on strong risk aversion, 
 \cite{L05} and \cite{FS16} on law-invariant risk measures,  
\cite{PWW25} on antimonotonic convexity. 
For independent pairs, $L^\infty$-continuity is not sufficient because the laws of large numbers do not offer convergence in $L^\infty$. 
The law of large numbers for negatively dependent sequences on $L^p$ 
requires  classic techniques in stochastic order \citep{MS02,SS07} and a recent result on uniform integrability by \cite{LV13}.
We offer many (counter)examples that  carefully design law-invariant and continuous mappings that violate various versions of diversification  while satisfying or failing  risk aversion. These examples illustrate the necessity of our assumptions and the exact scope of each result, which justify the strictness of the single-direction implications in Figure \ref{fig:1}.

\section{Preferences and risk aversion}

Let $(\Omega, \mathcal F, \p)$ be an atomless probability space and $L^\infty $ be the set of essentially bounded random variables on this space. Almost-sure equal random variables are treated as identical.
Random variables in $L^\infty$ are interpreted as random payoffs in one period. Constant random variables  are identified with elements in $\R$. The $L^\infty$-norm of a random variable $X$ is given by
$$\Vert X \Vert_\infty= \inf \{x\in \R:\p( |X|>x)=0\},$$ which is the essential supremum of $|X|.$  
In the main part of the paper, we work with the domain $L^\infty$, which is the standard space in decision theory and risk measures.   The results can be generalized to $L^p$ with $p\in [1,\infty)$, the space of random variables with finite $p$-th moment, which we discuss in Section \ref{sec:L1}. 
  Let $\Delta_n$ be the standard simplex in $\R^n$.
All terms like ``increasing" in this paper are in the weak sense.

We write $X\laweq Y$ when two random variables (or random vectors) $X$  and $Y$ are identically distributed (ID). 
The decision maker's preferences are represented by
a transitive binary relation $\succsim$  on $L^\infty$, called a preference relation, with    strict part $\succ$  and    symmetric part $\simeq$. 
A \emph{risk preference} $\succsim$  is a preference relation satisfying the following two  standard properties. 
\begin{enumerate} 
\item[(a)]  Law invariance: $X \laweq  Y \implies X \simeq Y$ for all $X ,Y\in L^\infty$,  
\item[(b)]  Upper semicontinuity: the set 
$\left\{ Y \in L^\infty: Y \succsim X \right\}$ is closed with respect to $L^\infty$-norm for each $X\in L^\infty$. 
\end{enumerate}

If in (b),  the set $\left\{ Y \in L^\infty: X \succsim Y \right\}$ is also closed, then $\succsim$ is \emph{continuous}.  
 Throughout, continuity is with respect to $L^\infty$-norm when not specified otherwise.   
Virtually all decision models satisfy this form of continuity. 
We do not assume completeness of $\succsim$ (each pair is comparable by $\succsim$) 
or monotonicity ($X\ge Y$ implies $X\succsim Y$). 
 This allows for incomplete and nonmonotone preferences, such as the  mean-variance preferences of \cite{M52}, that is, 
\begin{equation}\label{eq:MV}
X \succsim Y ~\iff~ \E[X]\ge \E[Y] \mbox{ and } \var(X) \le \var(Y).
\end{equation}
In all results, we do not assume any particular decision model for the risk preferences.

In many financial applications, the preference relation 
  $\succsim$ is represented  by 
  a utility functional $\mathcal U$ on $L^\infty$, that is,
\begin{equation}
\label{eq:utility}
X \succsim Y ~\iff~ \mathcal U(X)\ge \mathcal U(Y),
\end{equation}
 or 
  a risk measure $\rho$ on $L^\infty$ (with  a  sign flip), that is, $ 
X \succsim Y  \iff  \rho(-X)\le \rho(-Y).$  
The input of the risk measure is  $-X$, interpreted as the potential loss/gain from the payoff $X$,  following the convention of \cite{MFE15}.
 With \eqref{eq:utility}, property (a)  of $\succsim$ translates into law invariance of $\mathcal U$, i.e., $ X \laweq Y$ implies $\mathcal U(X)=\mathcal U(Y)$, and property  (b) translates into 
 the  upper semicontinuity of $\mathcal U$. These are standard properties and satisfied by common  utility functionals and risk measures.

For some results, we need a stronger notion of continuity, called compact continuity (\citealp{CM95,CL07}).
We say that 
 a sequence $(X_n)_{n\in \N}$ of random variables converges to $X$ in \emph{bounded convergence} if $(X_n)_{n\in \N}$ is uniformly bounded and  $X_n\to X$ almost surely. For law-invariant preference relations, it is safe to replace almost sure convergence here with convergence in distribution. 
\begin{enumerate}  
\item[(c)]  Compact continuity: the sets 
$\left\{ Y \in L^\infty: Y \succsim X \right\}$ 
and 
$\left\{ Y \in L^\infty: X \succsim Y \right\}$ 
are closed with respect to bounded convergence for each $X\in L^\infty$.
\end{enumerate}
Compact upper semicontinuity
is defined analogously. 
Compact (semi)continuity is stronger than $L^\infty$-(semi)continuity. For instance, 
denote by   
$Q_X$ the left quantile function of a random variable $X$, that is, $Q_X(t)=\inf\{x\in \R: \p(X\le x)\ge t\}$ for $t\in (0,1)$. The quantile mapping $X\mapsto Q_X(t)$ for any $t\in(0,1)$ is $L^\infty$-continuous but not compact continuous; another such example is the essential supremum functional $X\mapsto \esssup X$.

Next, we introduce notions of risk aversion. 
First, we need the concave order between two random variables $X, Y \in L^\infty$, written as $X \ge_{\rm cv} Y$, when 
\begin{equation*}
    \E[u(X)] \ge \E[u(Y)] \quad \text{for all concave} \ u:\R \to \R.
\end{equation*} 
For technical treatments on the concave order and its variants, see \cite{SS07}. In risk management, it is common to use the convex order, which is the reverse relation  of the concave order, that is, $X\ge_{\rm cx} Y \iff X\le_{\rm cv}Y$.

The weak and strong notions of risk aversion are defined next. For various notions of risk aversion in popular decision models and their characterization, see \cite{C95} and \cite{SZ08}.

\begin{definition}
A risk preference $\succsim$ exhibits \textit{weak risk aversion} if for $X \in L^\infty$,
\begin{equation*}
    \E[X] \succsim X.
\end{equation*}
A risk preference $\succsim$ exhibits \textit{strong risk aversion} if for $X,Y\in L^\infty$, 
\begin{equation*}
    X \ge_{\rm cv} Y \implies X \succsim Y.
\end{equation*}
Weak and strong notions of \textit{risk seeking} are defined by replacing $\succsim$ with $\precsim$ in the above implications, respectively. 
\textit{Risk neutrality} means $\E[X]\simeq X$ for all $X\in L^\infty$.
\end{definition}

It is straightforward to see that strong risk aversion implies weak risk aversion, and risk neutrality is equivalent to both   (either weak or strong) risk aversion and   risk seeking.
In the expected utility (EU) model, each of weak risk aversion and strong risk aversion is equivalent to a concave utility function. In the dual utility model of \cite{Y87}, weak risk aversion is strictly weaker than strong risk aversion. Incomplete and non-monotone preferences can exhibit risk aversion; for instance, \eqref{eq:MV} exhibits strong risk aversion.


\section{Diversification and dependence}

We first introduce a few notions of dependence that are important in statistical modeling. They will be essential in our formulation of diversification. 
\begin{enumerate}[label=(\alph*)]
\item 
 A pair $(X,Y)$ of random variables is {\it comonotonic}  if
\begin{equation*}
    (X(\omega)-X(\omega'))(Y(\omega)-Y(\omega')) \ge 0 \quad \text{for }
    (\omega, \omega') \in \Omega^2,\mbox{~$\p\times \p$-a.s.}
\end{equation*} 
\item 
A pair $(X,Y)$ is \emph{antimonotonic} (also called \emph{anticomonotonic}, or \emph{counter-monotonic}) if $(X,-Y)$ is comonotonic. 
\item 
A pair
 $(X,Y)$  is {\it exchangeable} 
 if $(X,Y)\laweq (Y,X)$.
 
\end{enumerate}
Comonotonicity describes the strongest form of positive dependence, 
whereas antimonotonicity describes the strongest form of negative dependence.
An exchangeable pair is necessarily ID. For ID pairs, all of comonotonicity,   independence, and antimonotonicity are special cases of exchangeability.  For a general treatment on these dependence concepts and their applications in finance, see \cite{J97} and \cite{DDGKV02}.

We now define diversification in a similar way to \cite{D89}, with the difference that we will restrict the random payoffs at comparison to those satisfying certain conditions specified by a class  $\mathcal X \subseteq (L^\infty)^2$ of   pairs of random variables.
\begin{definition}\label{def:main}
For  $\mathcal X \subseteq (L^\infty)^2$,   a risk preference $ \succsim$ exhibits 
 \emph{diversification on $\X$} if  
\begin{equation}\label{eq:QC}
    X \simeq Y \implies \lambda X + (1-\lambda)Y \succsim  Y \mbox{ for all $\lambda \in [0,1]$},
\end{equation}
 and for all $(X,Y)\in \mathcal X$.
 \end{definition}

 We use natural language to describe the class $\mathcal X$. For instance,  we   say ``diversification on antimonotonic and ID pairs",   meaning that  \eqref{eq:QC} holds for $(X,Y)$ that satisfy both antimonotonicity and ID. 
 When $\X=(L^\infty)^2$, we simply say ``diversification on all pairs".

 \cite{D89} formulated diversification on an arbitrary number of random payoffs, that is, 
 $$
   n\in \N,~ X_1\simeq   \dots\simeq X_n \implies \sum_{i=1}^n \lambda _i X_i \succsim  X_1  \mbox{~for all $(\lambda_1,\dots,\lambda_n )\in \Delta_n$},
 $$
 where $\Delta_n$ is the standard simplex in $\R^n$. 
 Our formulation  \eqref{eq:QC}  only involves pairs of payoffs in a set $\X$, thus a weaker requirement in general; some conditions on more than two payoffs are indirectly imposed through transitivity 
 of $\succsim$. 
A slightly stronger formulation than \eqref{eq:QC}  is 
\begin{align}
    \label{eq:QC2}
    X \succsim Y \implies \lambda X + (1-\lambda)Y \succsim  Y \mbox{ for all $\lambda \in [0,1]$},
\end{align} 
  and under mild conditions the two formulations are equivalent (e.g., \citealp{CT02}). 
The property in \eqref{eq:QC2} for all pairs $(X,Y)$ is called \emph{convexity}, \emph{concavity}, \emph{quasi-convexity}, or \emph{quasi-concavity} of $\succsim$ by different authors.
In the context of risk measures, \eqref{eq:QC2} becomes 
\begin{equation}
    \label{eq:QC-riskmeasure}
\rho(\lambda X + (1-\lambda)Y) \le \max\{\rho(X),\rho(Y)\},~~~X,Y\in L^\infty, ~\lambda\in [0,1],
\end{equation}
which is called the quasi-convexity of $\rho$, and is well studied by \cite{CMMM11} and \cite{DK13}.\footnote{A \emph{monetary} risk measure (\citealp{FS16}) is a mapping $\rho:L^\infty\to\R$ that satisfies  \emph{monotonicity}: $\rho(X)\ge \rho(Y)$ if $X\ge Y $, 
and \emph{cash additivity}: $ \rho(X+c)=\rho(X)+c$ for $c\in \R$ and $X\in L^\infty$. 
For monetary risk measures, quasi-convexity is  equivalent to the usual convexity. All law-invariant convex and monetary risk measures, as well as their maximum, minimum, and convex combinations,
 exhibit strong risk aversion \citep[Proposition 3.2]{MW20}.} 
 The properties  \eqref{eq:QC2}--\eqref{eq:QC-riskmeasure} are also  essential to the formulation of  aspirational preferences and the associated risk measures; see \cite{BGS2012}.

 

\section{Relations between diversification and risk aversion}
\label{sec:main}

Diversification is closely related to risk aversion, as already observed by \cite{D89}. In this section we explore  how imposing specific dependence structures  in diversification  affects risk aversion.

\subsection{Comonotonic pairs}
Our first   observation is that diversification for comonotonic pairs does not lead to any notion of risk aversion. Intuitively,  $X$ and $Y$ in a comonotonic pair do not  hedge each other in the portfolio $\lambda X+(1-\lambda) Y$.  
If $(X,Y)$ is comonotonic, then $$Q_{\lambda X+(1-\lambda)Y}=\lambda Q_X+ (1-\lambda) Q_Y.$$
Therefore, the left quantile is affine on comonotonic pairs, although quantiles do not exhibit risk aversion or risk seeking in general; see \cite{MFE15} for more discussions on comonotonicity and using   quantiles as   risk measures in finance.
Hence, diversification on comonotonic pairs is not directly related to hedging considerations and it does not
force  the decision maker to be risk averse. The following proposition makes this simple point clear. It further illustrates that a risk preference  can exhibit  both  diversification on comonotonic pairs and \emph{strict strong risk seeking}, that is, 
\begin{equation}\label{eq:ssrs}
\mbox{ for all $X,Y$ with  $X\not\laweq Y$}, ~
  X \ge_{\rm cv} Y  \implies Y \succ  X.
\end{equation}

\begin{proposition}\label{prop:1}
For a risk preference, diversification on comonotonic pairs 
does not imply weak risk aversion. Indeed, the risk preference $\succsim$ represented by $U$ via  \eqref{eq:utility} with
$$
\mathcal U(X) = \int_0^1 g(t)  Q_X(t) \d t, ~~X\in L^\infty, \mbox {~~for any strictly increasing function $g$,}
$$
exhibits diversification on comonotonic pairs  and strict strong risk seeking in \eqref{eq:ssrs}.
\end{proposition} 
  \begin{proof} 
It suffices to show the second statement. 
 Note that $\succsim$ belongs to the dual utility of \cite{Y87} with a strictly concave weighting function. As a common property of the dual utility functional,  $\mathcal U$ is affine on comonotonic pairs, and hence diversification on comonotonic pairs holds. We can check that it also  satisfies \eqref{eq:ssrs};   a precise statement of this fact can be found in  \citet[Corollary 1]{LLW25}.
  \end{proof}

\cite{CT02} showed that in the EU model,  diversification on comonotonic pairs is equivalent  to  both diversification on all pairs and  strong risk aversion. Combined with Proposition \ref{prop:1}, this highlights the coarse nature of the EU model in its treatment of  diversification. 

\subsection{Antimonotonic pairs}
\label{sec:anti}
In contrast to the negative result in 
Proposition \ref{prop:1}, we present a positive result that diversification on antimonotonic pairs, which is intuitively plausible,
 has a 
 normatively appealing consequence, that is, weak risk aversion.

\begin{theorem}\label{th:1}
For a risk preference, diversification  on antimonotonic and ID pairs implies  weak  risk aversion,
and it is implied by strong risk aversion.
Both implications are in general strict.
\end{theorem} 
  \begin{proof}   
We first show the implication from diversification on antimonotonic and ID  pairs to weak risk aversion. 
Let $X \in L^\infty$ and  $U$ be uniformly distributed  on $[0,1]$. Define
\begin{equation*}
    X_0^{(1)} = Q_X(U), \quad X_0^{(2)} = Q_X(1-U), 
\quad  \mbox{and} \quad 
    X_1 = \frac{X_0^{(1)} + X_0^{(2)}}{2}.
\end{equation*}
Clearly,   $ X \laweq X_0^{(1)} \laweq X_0^{(2)}$. 
Further, by construction, $X_0^{(1)}$ and $X_0^{(2)}$ are anti-comonotonic. By diversification on antimonotonic pairs and law invariance of $\succsim$,  we have 
\begin{equation*}
    X_1  =  \frac{1}{2}X_0^{(1)} + \frac{1}{2}X_0^{(2)} \succsim 
    X_0^{(1)} \simeq X,
\end{equation*}
and
\begin{equation*}
    \E[X_1] = \frac{1}{2}\E \left[X_0^{(1)} \right] + \frac{1}{2}\E \left[ X_0^{(2)} \right] = \E[X].
\end{equation*}
Inductively, for $n \in \N$, we can construct  
\begin{equation*}
    X_{n}^{(1)} = Q_{X_{n}} (U),  \quad X_{n}^{(2)} = Q_{X_{n}} (1-U), \quad \text{and} \quad
    X_{n+1} = \frac{X_{n}^{(1)} +  X_{n}^{(2)}}{2} . 
\end{equation*}
Following the same arguments, we have 
\begin{equation*}
    X_n \succsim X_{n-1} \succsim \cdots \succsim X_1 
    \succsim X \ \ 
    \text{and} \ \ \E[X_n] = \E[X].
\end{equation*}
For $ n \in \N$, let $R_n = \esssup X_n - \essinf X_n$, where for any random variable $Z$, $\esssup Z$ is its essential supremum and $\essinf Z$ is its essential infinimum.  Clearly, $$\essinf X_n \le \E[X_n] = \E[X] \le \esssup X_n.$$ 
Hence,  
 \begin{equation*}
    |X_n  - \E[X]| \le \esssup X_n - \essinf X_n = R_n  \quad \mbox{$\p$-a.s.}
 \end{equation*}
and thus $
    \| X_n - \E[X]\|_\infty \le R_n.
 $
  Lemma 3.1 of \cite{HWWW24} gives \begin{equation}
  \label{eq:Rdecay}
      R_{n+1}\le \frac{R_{n}}2~~~\mbox{ for $n\ge0$.}
  \end{equation}
We here give a short self-contained proof of \eqref{eq:Rdecay}. 
Let $m=Q_{X_n}(1/2)$. 
  When $U \le 1/2$, it is  $Q_{X_n}(U) \le  m \le Q_{X_n}(1-U)$.
  When  $U > 1/2$,  it is  $Q_{X_n}(U) \ge  m \ge Q_{X_n}(1-U)$.
  In both cases, we have  
 $$
    \frac{\essinf X_n + m}{2} \le \frac{Q_{X_n} (U) + Q_{X_n} (1-U)}{2} \le \frac{m + \esssup X_n}{2}.
 $$
Therefore, $X_{n+1}$ is between $(\essinf X_n + m)/2$ and $(\esssup X_n + m)/2$, thus showing \eqref{eq:Rdecay}. As a consequence of this,
 $
    R_n \le  {R_0}/{2^n} \to 0$ as $n \to \infty$. 
Therefore,
\begin{equation*}
    0 \le \lim_{n \to \infty} \| X_n - \E[X]\|_\infty \le \lim_{n \to \infty} R_n =0.
\end{equation*} By the upper semicontinuity of $\succsim$, we conclude
\begin{equation*}
    X_n \succsim X \ \ \text{for all } \ n \in \N ~\implies~ \E[X] \succsim X.
\end{equation*} 

To show strong risk aversion implies diversification on   ID pairs (antimonotonic or not), it suffices to note that for any $X\laweq Y$ and $\lambda\in [0,1]$, we have $\lambda X+ (1-\lambda) Y\ge _{\rm cv} X$, following from Jensen's inequality. 
The strictness of both implications
is justified by  Example \ref{ex:weak-strange} and Remark \ref{rem:anti}.
  \end{proof}



\begin{example}[Weak risk aversion $\centernot\implies$ diversification on AM and ID]\label{ex:weak-strange}
  Let the risk preference $\succsim$ be given by    $$
X \succsim Y ~\iff~\mathcal U(X) \ge \mathcal U(Y) ,
 $$ 
 where $\mathcal U(Z)=  \E[Z] -\var(Z)|2-\var(Z)|$ for $Z\in L^\infty$. 
 It is clear that  $\succsim$  exhibits weak risk aversion because $\mathcal U(X)\le \E[X]= \mathcal U(\E[X])$. 
 Let $A,B,C$ form a partition of $\Omega$ with equal probability,  $X=3\id_{A}$, and $Y=3\id_B$. Clearly, $(X,Y)$ is an antimonotonic and ID pair. Let $Z=(X+Y)/2$.  
 We can compute $\E[X]= 1$,  $\var(X)=2$, $\E[Z]=1$, and $\var(Z)= 1/2$.
 Therefore, 
 $
\mathcal U(X)
 =1  >1/4=  \mathcal U(Z), 
 $
 violating diversification on antimonotonic and ID pairs.
\end{example}

\begin{remark}[Diversification on AM $\centernot\implies$ strong risk aversion]
\label{rem:anti}
\cite{ACV21} showed that, for preferences represented by Choquet integrals, quasi-convexity on antimonotonic pairs is strictly weaker than convexity. Applying this to the dual utility model of \cite{Y87}, we get that diversification for antimonotonic pairs does not imply strong risk aversion. 
 \end{remark}



\begin{example}[Strong risk aversion $\centernot\implies$ diversification on AM]
\label{ex:mean-variance}
    Let the risk preference $\succsim$   be given by    $$
X \succsim Y ~\iff~  \E[ X] - (\var(X))^{1/4} \ge \E[Y]-  (\var(Y) )^{1/4}.
 $$ 
 It is straightforward to check that  $\succsim$  exhibits strong risk aversion. Let $X=1$, $Y$ take values  $1$ and $3$ with equal probability, and $Z=(X+Y)/2$. Note that $(X,Y)$ is antimonotonic since $X$ is a constant. We have $\E[Y]=2$,    $\var(Y)=1$, $\E[Z]=3/2$, and $\var(Z)=1/4$. 
 Hence, $X\simeq Y$ and 
 $\E[Z]-(\var(Z))^{1/4}=3/2-(1/2)^{1/2} <1=\E[X]-(\var(X))^{1/4}$, showing that $Z\prec X$, violating diversification on antimonotonic pairs.
 Nevertheless, diversification on antimonotonic and ID pairs holds by Theorem \ref{th:1}.
\end{example}

\begin{remark}
    \label{rem:CL07} 
     
In order to get an equivalent characterization of weak risk aversion, one needs to exclusively restrict attention to comparisons between a constant and a random variable.  \citet[Theorem 3.1]{CL07} show that weak risk aversion is equivalent to
       \begin{equation}
    \label{eq:sure1} n\in \N,~ X_1\simeq   \dots\simeq X_n,~(\lambda_1,\dots,\lambda_n)\in \Delta_n,~ \sum_{i=1}^n \lambda _i X_i \in \R  \implies \sum_{i=1}^n \lambda _i X_i \succsim  X_1,
    \end{equation} 
    under additional conditions: completeness, monotonicity, and compact continuity. \citet[Theorem 1]{MMWW25} characterized weak risk aversion via
    \begin{equation}
    \label{eq:MMWW}
Y\laweq Z \mbox{ and $X+Y\in \R$} \implies X+Y\succsim X+Z,
\end{equation}
with no additional assumptions on $\succsim$ other than law invariance and transitivity. 
None of \eqref{eq:sure1} (even restricted to $n=2$) and \eqref{eq:MMWW} is compatible with Definition \ref{def:main}.
\end{remark}

\subsection{Exchangeable pairs}

We next focus on diversification on exchangeable pairs, which turns out to be equivalent to diversification on ID pairs.  
\begin{theorem}\label{th:1ex}
For a risk preference, the following are equivalent: 
\begin{enumerate}[label=(\roman*)]  \item
    strong risk aversion;
    \item diversification  on ID pairs;
    \item  diversification  on exchangeable pairs. 
\end{enumerate} 
\end{theorem} 
  \begin{proof} 
\underline{(i)$\Rightarrow$(ii)}:  
  Strong risk aversion implies diversification on ID pairs, as we see in the proof of Theorem \ref{th:1}. 
\underline{(ii)$\Rightarrow$(iii)}: This follows by definition.  
We will prove the most involved direction, \underline{(iii)$\Rightarrow$(i)}, below. 

Take $X,Y\in L^\infty$ with $X \ge_{\rm cv} Y$. By Strassen's Theorem (\citealp{S65}), there exists $(X',Y')$  such that $X' \laweq  X$, $Y' \laweq  Y$, and $\E[Y'  \mid X']=X'$. By law invariance of $\succsim$, it suffices to show $X' \succsim Y'$. Therefore, it is without loss of generality to assume $\E[Y|X]=X$. Further, since the risk preference is law invariant,  it does not lose generality to assume that there exists a sequence $(U_n)_{n\in \N}$ of independent and ID uniformly distributed random variables on $[0,1]$ independent of $X$.

We first analyze the case when $X$ takes   values  in a finite set $\mathcal S$.
Let $Z_0=Y-X$. Inductively for $n\ge 0$, we define the following quantities. 
Define the function $$Q_n(s,t)=\inf\{z\in \R: \p(Z_n\le z\mid X=s)\ge t\},~~~t\in (0,1),~s\in \mathcal S,$$
which is the conditional quantile of $Z_n$ given $X=s$.  
Let
\begin{equation*}
    Z_n^{(1)}  = Q_n(X , U_n ), \quad 
    Z_n^{(2)}  = Q_n(X, 1-U_n ), \quad \mbox{and} \quad
    Z_{n+1} = \frac{Z_n^{(1)} + Z_n^{(2)}}{2}.
\end{equation*}
Further, set  $Y_n^{(i)}=X+Z_n^{(i)}$ for $i\in \{1,2\}$ 
and $Y_n=X+Z_n$.
It is clear that for $n \in \N$, $Y_{n+1}= (Y_n^{(1)}+Y_n^{(2)})/2$.

By independence between $U_n $ and $X$, we have that  $Z_n^{(1)}$, $Z_n^{(2)}$, and $Z_n$ have the same conditional distribution  on  $X=s$ for each $s\in \mathcal S$, because they have the same conditional quantile function. 
This implies 
 $
  Y_n^{(1)} \laweq  Y_n^{(2)} \laweq  Y_n,
 $
 and moreover, $  (Y_n^{(1)},  Y_n^{(2)})$ is exchangeable. 
Therefore, 
\begin{equation*}
  \E[Z_{n+1}\mid X]
  = \frac{1}{2} \E[Z_n^{(1)}\mid X]+\frac{1}{2} \E[Z_n^{(2)}\mid X]
  = \E[Z_n\mid X].
\end{equation*}
By induction from $\E[Z_0\mid X]=0$ we get $\E[Z_n\mid X]=0$ for all $n$. 

Note that $\Vert Y_n-X\Vert_\infty=\Vert Z_n\Vert_\infty$. Using the same argument as in part (i) for \eqref{eq:Rdecay}, we get that the length of the range of $Z_n$ conditionally on $X=s$ for each $s\in \mathcal S$ shrinks to $0$. Since $\mathcal S$ is a finite set, this implies $\Vert Z_n\Vert_\infty\to 0$ as $n\to\infty$.  
Because    $(Y_n^{(1)},  Y_n^{(2)})$ is exchangeable  and $Y_{n+1}=({Y_n^{(1)}+Y_n^{(2)}})/{2}$, diversification on exchangeable pairs implies
 $
  Y_{n+1} \succsim  Y_n $ for all $n\in\N.
 $
By the upper semicontinuity of $\succsim$ and $\|Y_n-X\|_\infty = \Vert Z_n\Vert_\infty \to 0$ as $n\to\infty$, we obtain
 $
  X  \succsim Y_0 = X+ Z_0 = Y,
 $ showing strong risk aversion. 

For general $X$ that may take infinitely many values, we rely on the following simple lemma. 
\begin{lemma}\label{lem:1}
For $X \in L^\infty$, there exists a sequence of finitely valued random variables $(X_n)_{n\in \N}$ such that
 $$
  \|X_n - X\|_\infty \to 0 
  \quad\text{and}\quad 
  X_n \ge_{\mathrm{cv}} X \ \text{for all } n\in \N.
 $$
\end{lemma}

 \begin{proof}[Proof of the lemma.] 
For each $n\in\mathbb{N}$, let $\mathcal{G}_n$ be the finite $\sigma$-algebra generated by $\{X\in I_n^k\}_{k=1,\dots,g_n}$, where $(I_n^1,\dots,I_n^{g_n})$ is a finite partition of the support of $X$ into intervals of length at most $2^{-n}$, and define
 $$
  X_n = \mathbb{E}[X \mid \mathcal{G}_n].
 $$
Then $X_n$ is finitely valued and $\|X_n - X\|_\infty \to 0$.  Moreover,  $X_n \ge_{\rm cv} X$ for all $n\in \N$ by the conditional Jensen's inequality.
  \end{proof}

Now we continue to prove Theorem \ref{th:1ex}. Let the sequence $(X_n)_{n\in\N}$ be as in Lemma \ref{lem:1}. Transitivity of the concave order gives $X_n\ge_{\rm cv} X\ge_{\rm cv} Y.$
Using the obtained result on finitely-valued random variables, we conclude $X_n\succsim Y$ for each $n$. Applying the upper semicontinuity of $\succsim$ to the above relation with $  \|X_n - X\|_\infty \to 0$, we get $X\succsim Y$, thus showing the desired statement of strong risk aversion. 
  \end{proof} 

The most important direction in Theorem \ref{th:1ex} 
is (iii)$\Rightarrow$(i), and it 
generalizes 
several results in the literature. 
\citet[Theorem 4.2]{CL07} obtained that, under completeness, strict monotonicity, and compact continuity (essential to their proof), 
diversification on ID pairs is equivalent to strong risk aversion.
Our result relaxes ID pairs to  exchangeable pairs, removes completeness and monotonicity, and uses   $L^\infty$-upper semicontinuity that is weaker than   compact continuity. 
In the risk measure literature, $L^\infty$-continuity is common and satisfied by all monetary risk measures.
Theorem \ref{th:1ex} thus generalizes a classic result in the risk measure literature: A law-invariant convex and monetary risk measure on $L^\infty$ with the Fatou property exhibits strong risk aversion  \citep[Corollary 4.65]{FS16}.\footnote{The result  was shown for coherent risk measures by \cite{L05}. The Fatou property can be omitted, which is first shown  by \cite{JST06} and then strengthened by \citet[Theorem 30]{D12}.}  
Since convex risk measures satisfy \eqref{eq:QC-riskmeasure},  the above result is a special case of Theorem \ref{th:1ex}.     We present  a  corollary here, stronger than the existing results on risk measures, using the convex order.
\begin{corollary}
\label{coro:RM-convex}
    A law-invariant mapping $\rho:L^\infty\to\R$ satisfying  lower semicontinuity and
\begin{equation}\label{eq:rho-convex}
\rho(\lambda X+(1-\lambda)Y) \le \rho(X) \mbox{ for all $X,Y\in L^\infty$ with $(X,Y)\laweq (Y,X)$ and $\lambda \in [0,1]$}
\end{equation}
    is increasing in the convex order.  
\end{corollary}
\begin{remark}
A  simple sufficient condition for  $\rho:L^\infty \to \R$ to satisfy both law invariance and \eqref{eq:rho-convex} is
  $
\rho(\lambda X+(1-\lambda)Y) \le \rho(X) \mbox{ for all $X,Y\in L^\infty$ with $X\laweq  Y$ and $\lambda \in [0,1]$}.
$
 
\end{remark}

For the EU model, weak and strong notions of risk aversion coincide,  and hence Theorems \ref{th:1}--\ref{th:1ex} together imply that diversification for antimonotonic pairs is equivalent to the concavity of the utility function, stated in   \citet[Theorem 7]{PWW25}.    


\subsection{Independent pairs}
\label{sec:indi}

We now consider diversification on independent pairs, whose implications on risk aversion depend on the continuity assumptions, as we will see from the results in this section. 

\begin{proposition}\label{prop:2}
For a risk preference, diversification on independent pairs 
does not imply weak risk aversion. Indeed, the risk preference $\succsim$ represented by $\mathcal U$ via \eqref{eq:utility} with 
$$\mathcal U(X)=\esssup X,~~~X\in L^\infty$$ 
exhibits diversification on independent pairs  and   strong risk seeking.
\end{proposition} 
  \begin{proof}  
It is clear that $\succsim$ exhibits strong risk seeking, because $X\ge_{\rm cv}Y$ implies $\esssup X\le \esssup Y$ and thus $X\precsim Y$.  
For $X,Y$ independent with $X\simeq Y$,  we have 
$$\esssup (\lambda X+(1-\lambda)Y) = \lambda \esssup  X+(1-\lambda)\esssup Y = \esssup Y,$$ and hence
$ \lambda X+(1-\lambda)Y\simeq Y$.
Therefore,  $\succsim$ exhibits diversification on independent pairs. 
  \end{proof}

\begin{remark}
    \label{rem:indi}
 Example \ref{ex:mean-variance} 
 illustrates that  strong risk aversion does not imply 
 diversification on independent  pairs,  noting that $(X,Y)$ in that example is independent. 
Together with Proposition \ref{prop:2}, we see that these two concepts are incomparable. 
\end{remark}

\begin{remark}
As we see in Proposition \ref{prop:1},  diversification on comonotonic pairs is compatible with strict strong risk seeking in \eqref{eq:ssrs}.
In contrast, diversification on independent pairs conflicts with strict strong risk seeking. 
To see this, take $X$ and $Y$ independent and both following a uniform distribution on $[0,1]$. 
Diversification on independent pairs would imply $X/2+Y/2 \succsim Y$, and strict strong risk seeking  would imply  $X/2+Y/2 \prec Y$, conflicting each other. 
That is why in Proposition \ref{prop:2} we can only state strong risk aversion but not the strict version.
\end{remark}

Our next result connects diversification on independent and ID pairs to strong risk aversion under compact  upper semicontinuity, which  is stronger than $L^\infty$-upper semicontinuity   and weaker than $L^p$-upper semicontinuity for any $p\in [1,\infty)$. 
\begin{theorem}
\label{th:2}
For a compact upper semicontinuous risk preference,   diversification  on independent and ID pairs implies  weak  risk aversion,
and it is implied by strong risk aversion.
Both implications are in general strict.
\end{theorem}
  \begin{proof} 
The implication 
that strong risk aversion implies diversification on independent and ID pairs
follows from Theorem \ref{th:1ex}. We now show that diversification on independent and ID pairs implies weak risk aversion.
Let $X\in L^\infty$ and  $(X_n)_{n\in \N}$ be a sequence of independent   random variables with the same distribution  as $X$.
Write $S_n=\sum_{i=1}^{2^n}X_i$ for $n\in \N$. 
By the law of large numbers, we have that $ S_n /2^n\to \E[X]$ almost surely. Note that $S_n/2^n$ is uniformly bounded, so  $ S_n /2^n\to \E[X]$ in bounded convergence.
Diversification on independent and ID pairs implies $S_{n+1}\succsim S_n$ for $n\in \N$. Transitivity and compact upper semicontinuity of 
$\succsim $ give  
$\E[X]\succsim S_n \succsim
\dots \succsim S_1 \simeq X$.
Therefore, weak risk aversion holds.  
Examples demonstrating that the converses of the two implications fail are given in Examples \ref{ex:not-strong}  and   \ref{ex:weak-strange2}, respectively. 
  \end{proof}

\begin{example}[Diversification on IN $\centernot\implies$ strong risk aversion]
\label{ex:not-strong}

Define 
$
\mathcal V(X)=\E[e^{2X}]/\E[e^{X}] 
$ for $X\in L^\infty$,
  and let the risk preference $\succsim$ be given by $
X\succeq Y  \Longleftrightarrow  \mathcal V(X)\le \mathcal V(Y).$ 
It is straightforward to check that $\succsim$ satisfies compact continuity. 
It also satisfies diversification on independent pairs by noting that 
$\mathcal V(\lambda X+(1-\lambda)Y)\le \mathcal V(X)^{\lambda}  \mathcal V(Y)^{1-\lambda}$ for $X,Y$ independent and $\lambda \in [0,1]$; this follows from standard calculus.   
Therefore, if $X\simeq Y$ and $X,Y$ independent, then $\mathcal V (\lambda X+(1-\lambda)Y)\le \mathcal V(Y)$. 
Finally, $\succsim$ does not exhibit strong risk aversion, with the counterexample $(X,Y)$ specified by $\p(X=1)=\p(X=-1)=\p(Y=1)=1/2$  and $\p(Y=-3/2)=\p(Y=-1/2)=1/4$, which satisfies $X\ge_{\rm cv} Y$ and $X\prec Y$. 
 \end{example}

\begin{example}[Weak risk aversion $\centernot\implies$ diversification on IN and ID]
\label{ex:weak-strange2}
Consider the risk preference $\succsim$ exhibiting weak risk aversion given in Example \ref{ex:weak-strange},  represented by the utility functional $\mathcal U(Z)=  \E[Z] -\var(Z)|2-\var(Z)|$ for $Z\in L^\infty$. 
It is clear that $\succsim$  is compact continuous. 
Let the  distribution of $X$ be the same as in Example \ref{ex:weak-strange}, that is, $\p(X=3)=1/3$ and $\p(X=0)=2/3$,  $X$ and $Y$ be independent and ID, and $Z=(X+Y)/2$.  
 We can compute $\E[X]= 1$,  $\var(X)=2$, $\E[Z]=1$, and $\var(Z)= 1$.
 Therefore, 
 $
\mathcal U(X)
 =1  > 0=  \mathcal U(Z), 
 $
 violating diversification on independent and ID pairs.  
 
\end{example}

\subsection{Strict single-directional implications in Figure \ref{fig:1}}
We now justify that the  single-direction implications
in
Figure \ref{fig:1} are strict in general, using the abbreviations therein.
 We have shown that diversification on AM pairs is incomparable to strong risk aversion (Remark \ref{rem:anti} and Example \ref{ex:mean-variance}),
and so is diversification on IN pairs (Remark \ref{rem:indi}). These observations and Theorem \ref{th:1ex} imply that the three notions in the second row of Figure \ref{fig:1} are incomparable, and hence diversification on all pairs is strictly stronger than each of them.  
The strictness of the implication from 
diversification on ID  pairs to  diversification on AM (resp.~IN)  and ID pairs follows from Theorems \ref{th:1} and \ref{th:1ex} (resp.~Theorems \ref{th:1} and \ref{th:2}).
The strictness of the implication from diversification on AM  (resp.~IN) and ID pairs to  weak risk aversion is given in
Theorem \ref{th:1}  (resp.~Theorem \ref{th:2}). 
The strictness of the implication from 
diversification on AM (resp.~IN) pairs  to diversification on AM (resp.~IN) and ID pairs is justified by the fact that the former is incomparable to strong risk aversion and the latter is implied by strong risk aversion. 
The strict implication from strong to weak risk aversion is well known.



\section{Neutrality}

The opposite side of risk aversion is risk seeking, and a combination of both is risk neutrality. 
Similarly, we can define the opposite of diversification preferences, and the corresponding neutrality.
\begin{definition} 
For  $\mathcal X \subseteq (L^\infty)^2$,   a risk preference $ \succsim$ exhibits 
 \emph{anti-diversification on $\X$} if  
\begin{equation}\label{eq:anti-QC}
    X \simeq Y \implies 
    X \succsim \lambda X + (1-\lambda)Y   \mbox{ for all $\lambda \in [0,1]$},
\end{equation}
 and for all $(X,Y)\in \mathcal X$.
A risk preference  exhibits \emph{diversification neutrality} if both diversification and 
 anti-diversification hold.

  \end{definition}
  Anti-diversification on different classes  describes situations in which the decision maker does not wish to diversify. 
  By applying our results to the reverse relation of $\succsim$, we can see that all results hold when we replace ``risk aversion" with ``risk seeking"
  and ``diversification" with ``anti-diversification". 
    Moreover,
 combining our main results in the previous section, we obtain the following equivalence between various forms of neutrality. We will involve an additional assumption of \emph{monotonicity (on constants)}: 
 $$
x \ge y  ~\implies~ x \succsim y \qquad \mbox{ for all $x,y\in \R$.}
 $$ 

  \begin{theorem}\label{th:3}
For a continuous risk preference $\succsim$, 
the following are equivalent:
\begin{enumerate}[label=(\roman*)]
    \item risk neutrality;   
    \item diversification neutrality on ID pairs; 
      \item diversification neutrality on exchangeable pairs; 
        \item diversification neutrality on antimonotonic and ID pairs.
    \end{enumerate}
If $\succsim$ is monotone, then each of the above is equivalent to 
\begin{enumerate}[label=(\roman*), resume]  
    \item diversification neutrality on all pairs;    
    \item diversification neutrality on antimonotonic pairs.  
\end{enumerate}
If $\succsim$  is monotone and compact continuous, then each of the above is equivalent to 
\begin{enumerate}[label=(\roman*), resume] 
    \item diversification neutrality on independent pairs;   
    \item  diversification neutrality on independent and ID  pairs.  
\end{enumerate}
      
  \end{theorem}
   \begin{proof}  \underline{(i)$\Rightarrow$(ii)}: Risk neutrality implies $\E[X]\simeq X$ for all $X\in L^\infty$. For $X\laweq Y$ and $\lambda \in[0,1]$, 
     we have
     $\lambda X+ (1-\lambda ) Y \simeq \E[ \lambda X+ (1-\lambda ) Y ]  
     = 
     \E[X] \simeq X,
     $
     and thus diversification neutrality on ID pairs holds. \underline{(ii)$\Rightarrow$(iii)$\Rightarrow$(iv)}: These follow by definition.  \underline{(iv)$\Rightarrow$(i)}: This follows by applying Theorem \ref{th:1} to both $\succsim$ and $\precsim$, and noting   that weak risk aversion  and weak risk seeking together imply risk neutrality.

         Next, assume   monotonicity. 
 \underline{(i)$\Rightarrow$(v)}: For $X\simeq Y$ and $\lambda \in[0,1]$ with $\E[X]\le \E[Y]$
     we have
     $$X\simeq \E[X] \le \lambda X+ (1-\lambda ) Y \simeq \E[ \lambda X+ (1-\lambda ) Y ]  
     \le 
     \E[Y] \simeq Y\simeq X,
     $$
     and by transitivity of $\succsim$ diversification neutrality on all pairs holds. 
\underline{(v)$\Rightarrow$(vi)$\Rightarrow$(iv)}: These  follow by definition.

         Finally, assume   monotonicity and compact continuity.  
\underline{(v)$\Rightarrow$(vii)$\Rightarrow$(viii)}: These follow by definition. 
 \underline{(viii)$\Rightarrow$(i):}  This follows by applying Theorem \ref{th:2} to both $\succsim$ and $\precsim$  and,  again, noting   that weak risk aversion  and weak risk seeking together imply risk neutrality.
   \end{proof}

If we assume
 \emph{strict monotonicity} for the risk preference $\succsim$, that is, 
 $$
x >y  ~\implies~ x \succ y \qquad \mbox{ for all $x,y\in \R$,}
 $$ 
 then  statements (i)--(vi) in  Theorem \ref{th:3} 
 are all equivalent to a representation of $\succsim$ by the mean, that is,
    $ 
    X\succsim Y \iff \E[X]\ge \E[Y].
    $   
 The next example shows that monotonicity cannot be removed from the implications (i)$\Rightarrow$(v) 
 and 
 (i)$\Rightarrow$(vi)
 in Theorem \ref{th:3}. 
 
 \begin{example}
The   risk  preference  $\succsim $ given by 
 $
X\succsim Y \iff (\E[X])^2 \ge (\E[Y])^2
 $
 exhibits risk neutrality but it is not monotone. 
 It does not satisfy diversification for antimonotonic pairs because  for $X$ with $\E[X]\ne 0$, we have
 $
 X\simeq -X 
 $ and $0=(X-X)/2\prec X$.
 Therefore, (i) in Theorem \ref{th:3} holds but neither (v) nor (vi) does.
  \end{example}

The risk preference represented by the essential supremum  in Proposition \ref{prop:2}  satisfies diversification neutrality on independent pairs. This  shows that the compact continuity assumed for the implication (viii)$\Rightarrow$(i)  
cannot be dispensed with.

\section{Extension to unbounded random variables}
\label{sec:L1}

In many financial applications concerning diversification, 
the payoffs of assets are not necessarily bounded; see the textbook \cite{MFE15} for discussions on the empirical evidence.
The natural domain to define the two forms of risk aversion is $L^1$, as both notions require integrability of the random payoffs to compare.

All our main results can be naturally extended to law-invariant preference relations $\succsim$ on $L^p$ for $p\in [1,\infty)$ with  similar proof techniques, but the $L^\infty$- and compact upper semicontinuity of $\succsim$ need to be strengthened to $L^p$-upper semicontinuity to accommodate convergence in the larger space. In this section, we show that the results in Theorems \ref{th:1}--\ref{th:3} hold in the $L^p$ setting under $L^p$-upper semicontinuity of $\succsim$, following  similar proof arguments with some manipulations.

For Theorem \ref{th:1}  in the $L^p$ setting,  we use the same construction of $(X_n)_{n\in \N}$ as in the proof for the case of $L^\infty$, and instead of $X_n \to \E[X]$ in $L^\infty$ we need to show $X_n\to \E[X]$ in $L^p$. This is guaranteed by Theorem \ref{th:lln} below. To prove Theorem \ref{th:lln}, we first present a standard result on the concave order and negative dependence.
We say that a pair $(X_1,X_2)$ of random variables is \emph{negatively quadrant dependent} (NQD, \citealp{L66})  if
$$
\p(X_1\le x_1, X_2\le x_2)\le \p(X_1\le x_1)\p(X_2\le x_2) \mbox{  for all $x_1,x_2\in \R$}. 
$$
Clearly, both independence and antimonotonicity belong to NQD, and indeed they  have the largest and smallest $\p(X_1\le x_1, X_2\le x_2)$ satisfying the above inequality.  
\begin{lemma}\label{lem:anti}
    For random variables $X_1,X_2,Y_1,Y_2\in L^1$ satisfying  $(X_1,X_2)$ NQD, 
      $(Y_1,Y_2)$ independent,   $X_1\ge_{\rm cv} Y_1$ 
       and $X_2\ge_{\rm cv} Y_2$, we have 
       $
       X_1+X_2\ge_{\rm cv} Y_1+Y_2.
       $
\end{lemma}
  \begin{proof} 
Take $X_1'\laweq X_1$ and  $X_2'\laweq X_2$ such that $(X_1',X_2')$ is independent. We  have 
       $$
       X_1+X_2\ge_{\rm cv} X_1'+X_2'\ge_{\rm cv} Y_1+Y_2,
       $$    
       where the first inequality follows from  the fact that for given marginal distributions, ordering in the bivariate distribution function implies the convex order of the sum \citep[Theorem 3.8.2]{MS02}, and the second inequality follows from the closure property of the concave order under convolution \citep[Theorem 3.A.12]{SS07}.  
  \end{proof}
The next result gives an $L^p$-law of large numbers for  negatively dependent sequences of ID random variables, which may be of some interest in probability theory.
\begin{theorem}\label{th:lln}
For  $X\in L^p$, let 
      $(X_n)_{n\in \N}$  be a  sequence satisfying  $X_0=X$ and for $n\in \N$,
    \begin{align*} X_n&=\frac{X_{n-1}^{(1)}+X_{n-1}^{(2)}}{2}, \mbox{ where $X_{n-1}^{(1)}\laweq X_{n-1}^{(2)}\laweq X_{n-1}$ and $(X_{n-1}^{(1)},X_{n-1}^{(2)})$ is NQD}. \end{align*}
    Then $X_n\to \E[X]$  in $L^p$.
\end{theorem}
\begin{remark}
We comment on a few special cases of Theorem \ref{th:lln}. The case  with independent $(X_{n-1}^{(1)},X_{n-1}^{(2)})$ is a version of the $L^p$-law of large numbers for independent and ID sequences   in $L^p$.  
    The construction of $(X_n)_{n\in \N}$ with antimonotonic  $(X_{n-1}^{(1)},X_{n-1}^{(2)})$ appears in the proof of Theorem \ref{th:1}.
We note that $(X_n)_{n\in \N}$ is only specified in terms of its marginal distributions, and hence we cannot expect $X_n\to \E[X]$ almost surely.
\end{remark}
  \begin{proof} [Proof of Theorem \ref{th:lln}]
We will compare $(X_n)_{n\in \N}$ with another sequence $(S_n)_{n\in \N}$ given by $S_n=\sum_{i=1}^{2^n}Y_i/2^{n}$ for $n\in \N$, where $(Y_n)_{n\in\N}$ is an independent and ID sequence with the same distribution as $X$.  
Because $X_0\laweq S_0$, we can apply Lemma \ref{lem:anti} to get $X_1\ge_{\rm cv}S_1$.  By induction on $n\in \N$ and using Lemma \ref{lem:anti} repeatedly, we get $X_n\ge_{\rm cv} S_n$ for all $n\in \N$. 
Next, let us check that $|S_n|^p$ is uniformly integrable. Note that  since $S_n\le_{\rm cx}X$ 
where $\le_{\rm cx}$ is the convex order, we have that 
$(S_n)_+^p \le_{\rm icx} X_+^p$,
where  $\le_{\rm icx}$ is the increasing convex order and $(x)_+=\max\{x,0\}$; see e.g., \citet[Theorem 4.A.15]{SS07}.
This implies that $((S_n)_+^p)_{n\in \N}$ is uniformly integrable by using \citet[Theorem 1]{LV13}.
By a symmetric argument, $((-S_n)_+^p)_{n\in \N}$ is also uniformly integrable. This shows $(|S_n|^p)_{n\in \N}$ is  uniformly integrable.  
By the strong law of large numbers, $S_n\to \E[X]$ almost surely. Using the uniform integrability of $(|S_n|^p)_{n\in \N}$ and $S_n\to \E[X]$, we get  $\E[|S_n-\E[X]|^p]\to 0$ by  \citet[Theorem 4.5.4]{C01}.
Since $x\mapsto |x-\E[X]|^p$ is convex, we have  $\E[|X_n-\E[X]|^p]\le  \E[|S_n-\E[X]|^p]\to 0$.
  \end{proof} 

 Theorem \ref{th:1} in the $L^p$ setting follows by using   Theorem \ref{th:lln} with antimonotonicity and the same proof arguments for the case of  $L^\infty$. 
  Theorem \ref{th:1ex} in the $L^p$ setting  follows from the a similar  argument, by using  Theorem \ref{th:lln} on the conditional distributions and replacing the $L^\infty$-approximation  in Lemma \ref{lem:1} with an $L^p$-approximation. We omit the details here. The proof of 
Theorem \ref{th:2} in the $L^p$  setting follows by applying Theorem \ref{th:lln} with independence
and and the same proof arguments for the case of $L^\infty$. The proof of Theorem \ref{th:3} in the $L^p$
 setting carries over verbatim.  

\section{Conclusion}

The results  in this paper show that one can recover rich information about risk attitudes from relatively modest diversification principles, provided they are formulated on economically meaningful classes of pairs such as antimonotonic, exchangeable, and independent risks. 
The main obtained  relations  are summarized in Figure~\ref{fig:1}. 
Especially, if a decision maker prefers to combine antimonotonic risks, as in hedging or purchasing insurance, then weak risk aversion can be deduced; if they prefer to combine exchangeable risks, as in pooling similar assets, then strong risk aversion can be deduced.
Our counterexamples highlight the limits of diversification as a diagnostic for risk aversion, and they underscore the role played by law invariance, continuity, and completeness assumptions in existing axiomatic frameworks.

\end{document}